\documentclass{patmorin}
\usepackage{pat}
\usepackage{paralist}
\usepackage{dsfont}  %
\usepackage[T1]{fontenc}
\usepackage[utf8]{inputenc}
\usepackage{skull}
\usepackage{paralist}
\usepackage{graphicx}
\usepackage[noend]{algorithmic}

\usepackage[normalem]{ulem}
\usepackage{cancel}
\usepackage{enumitem} 

\usepackage{todonotes}

\usepackage[longnamesfirst,numbers,sort&compress]{natbib}

\usepackage[mathlines]{lineno}
\newcommand*\patchAmsMathEnvironmentForLineno[1]{%
 \expandafter\let\csname old#1\expandafter\endcsname\csname #1\endcsname
 \expandafter\let\csname oldend#1\expandafter\endcsname\csname end#1\endcsname
 \renewenvironment{#1}%
    {\linenomath\csname old#1\endcsname}%
    {\csname oldend#1\endcsname\endlinenomath}}%
\newcommand*\patchBothAmsMathEnvironmentsForLineno[1]{%
 \patchAmsMathEnvironmentForLineno{#1}%
 \patchAmsMathEnvironmentForLineno{#1*}}%
\AtBeginDocument{%
\patchBothAmsMathEnvironmentsForLineno{equation}%
\patchBothAmsMathEnvironmentsForLineno{align}%
\patchBothAmsMathEnvironmentsForLineno{flalign}%
\patchBothAmsMathEnvironmentsForLineno{alignat}%
\patchBothAmsMathEnvironmentsForLineno{gather}%
\patchBothAmsMathEnvironmentsForLineno{multline}%
}

\DeclareMathOperator{\tw}{tw}
\DeclareMathOperator{\lca}{lca}

\title{\MakeUppercase{A Linear-Time Algorithm for Product Structure in Planar Graphs}\thanks{This research was partly funded by NSERC. This paper is the journal-length version of the conference paper \cite{bose.morin.ea:linear}. This version provides a much more detailed description of the algorithm that covers a number of technicalities and implementation details that were previously omitted. It also reports on open-source C/C++ implementations of the algorithm.}}
\author{%
  Prosenjit Bose\thanks{School of Computer Science, Carleton University}\qquad
  Kaya Gouin\footnotemark[2]\qquad
  Pat Morin\footnotemark[2]\qquad
  Saeed Odak\thanks{Department of Computer Science, Aalto University}}

\date{}

\usepackage{tabularx}

\begin{document}

\maketitle

\begin{abstract}
  The \emph{Product Structure Theorem} for planar graphs (Dujmović et al.\ \emph{JACM}, \textbf{67}(4):22) states that any planar graph is contained in the strong product of a planar $3$-tree, a path, and a $3$-cycle.  We give a simple linear-time algorithm for finding this partition as well as several related partitions.  This improves on the previous $O(n\log n)$ time algorithm for finding this partition (Morin.\ \emph{Algorithmica}, \textbf{85}(5):1544--1558). The algorithm is practical; we provide open-source C and C++ implementations that, even on modest hardware, can process graphs with millions of vertices in a matter of seconds. 
\end{abstract}

\section{Introduction}

For two graphs $G$ and $X$, the notation $G\subseteq X$ denotes that $G$ is isomorphic to some subgraph of $X$. 
The following \emph{planar product structure theorems} have recently been used as a key tool in resolving a number of longstanding open problems on planar graphs, including queue number \cite{dujmovic.joret.ea:planar}, nonrepetitive chromatic number \cite{dujmovic.esperet.ea:planar}, adjacency labelling \cite{dujmovic.esperet.ea:adjacency}, universal graphs \cite{esperet.joret.ea:sparse}, $p$-centered colouring \cite{debski.felsner.ea:improved}, and vertex ranking \cite{bose.dujmovic.ea:asymptotically}.\footnote{For graphs $G$ and $X$, an \emph{$X$-decomposition} of $G$ is a collection $\mathcal{X}:=(B_x:x\in V(X))$ of subsets of $V(G)$ called \emph{bags} indexed by the vertices of $X$ and such that
 \begin{inparaenum}[(i)]
     \item for each $v\in V(G)$, $X[\{x\in V(X):v\in B_x\}]$ is connected; and
     \item for each $vw\in E(G)$, there exists some $x\in V(X)$ such that $\{v,w\}\subseteq B_x$.
\end{inparaenum}
The \emph{width} of $\mathcal{X}$ is $\max\{|B_x|:x\in V(X)\}-1$. In the special case where $X$ is a tree, $\mathcal{X}$ is called a \emph{tree decomposition} of $G$. The \emph{treewidth} $\tw(G)$ of $G$ is the minimum width of any tree decomposition of $G$.}\footnote{
For two graphs $G_1$ and $G_2$, the \emph{strong graph product} of $G_1$ and $G_2$, denoted $G_1\boxtimes G_2$, is a graph whose vertex set is $V(G_1\boxtimes G_2):= V(G_1)\times V(G_2)$ and that contains an edge between distinct vertices $v=(v_1,v_2)$ and $w=(w_1,w_2)$ if and only if
\begin{inparaenum}[(i)]
    \item $v_1=w_1$ and $v_2w_2\in E(G_2)$;
    \item $v_2=w_2$ and $v_1w_1\in E(G_1)$; or
    \item $v_1w_1\in E(G_1)$ and $v_2w_2\in E(G_2)$.
\end{inparaenum}
}

\begin{thm}[\citet{dujmovic.joret.ea:planar, ueckerdt.wood.ea:improved}]\label{meta}
  For any planar graph $G$, there exists:
  \begin{compactenum}[(a)]
    \item \label{three_tree} a planar graph $H$ of treewidth at most $3$ and a path $P$ such that $G\subseteq H\boxtimes P\boxtimes K_3$ \cite{dujmovic.joret.ea:planar};
    \item \label{four_tree} a planar graph $H$ of treewidth at most $4$ and a path $P$ such that $G\subseteq H\boxtimes P\boxtimes K_2$; and
    \item \label{six_tree} a planar graph $H$ of treewidth at most $6$ and a path $P$ such that $G\subseteq H\boxtimes P$ \cite{ueckerdt.wood.ea:improved}.
  \end{compactenum}
\end{thm}

In the applications of \cref{meta}, the proofs are constructive and lead to algorithms whose running time is dominated by the time required to compute the relevant partition. Indeed, the proofs of each part of \cref{meta} lead to $O(n^2)$ time algorithms, as observed already by \citet{dujmovic.joret.ea:planar}.  \citet{morin:fast} later showed that there exists an $O(n\log n)$ time algorithm to find the partition in \cref{meta}.\ref{three_tree}.  In the current note, we show that there exists a linear time algorithm for finding each of the three partitions guaranteed by \cref{meta}.  This immediately gives an $O(n)$-time algorithm for each of the following problems on any $n$-vertex planar graph $G$:
\begin{compactitem}
  \item computing an $O(1)$-queue layout of $G$ \cite{dujmovic.joret.ea:planar};
  \item nonrepetitively vertex-colouring $G$ with $O(1)$ colours \cite{dujmovic.esperet.ea:planar};
  \item assigning $(1+o(1))\log n$-bit labels to the vertices of $G$ so that one can determine from the labels of vertices $v$ and $w$ whether or not $v$ and $w$ are adjacent in $G$ \cite{dujmovic.esperet.ea:adjacency};
  \item mapping the vertices of $G$ into a universal graph $U_n$ that has $n^{1+o(1)}$ vertices and edges so that any pair of vertices that are adjacent in $G$ maps to a pair of vertices that are adjacent in $U_n$ \cite{esperet.joret.ea:sparse};
  \item colouring the vertices of $G$ with $O(p^3\log p)$ colours so that each connected subgraph $H$ of $G$ contains a vertex whose colour is unique in $H$ or contains vertices of at least $p+1$ different colours \cite{debski.felsner.ea:improved}; and
  \item colouring the vertices of $G$ with $O(\log n/\log\log\log n)$ integers so that the maximum colour that appears on any path $P$ of length at most $\ell$ appears at exactly one vertex of $P$ (for any fixed $\ell\ge 2$) \cite{bose.dujmovic.ea:asymptotically}.
\end{compactitem}

In addition to planar graphs, product structure theorems similar to \cref{meta} exist for $k$-planar graphs, $h$-framed graphs, bounded-genus graphs, and graphs from apex-minor-free families \cite{dujmovic.joret.ea:planar,dujmovic.morin.ea:graph,bekos.dalozzo.ea:graph}.  The existence of each of these partitions relies at some point on \cref{meta}.  Thus, the algorithm presented here can also serve as an optimal subroutine in the computation of product structure partitions of graphs from these classes.  However, these larger graph classes also require additional machinery that (at the time of writing) makes finding these partitions largely impractical.\footnote{Possible exceptions here are cases in which a bounded-genus graph, $k$-planar graph, or $h$-framed graph is given along with its embedding.}

The remainder of this paper is organized as follows: \cref{prelims} presents some necessary background and notation.  \cref{tripod_decompositions} reviews the proof of \cref{meta}.\ref{three_tree}.  \cref{algorithm} presents the linear time algorithm for finding the partition in \cref{meta}.\ref{three_tree}.  \cref{variants} describes the algorithms for finding the partitions in \cref{meta}.\ref{four_tree} and \cref{meta}.\ref{six_tree}.  \Cref{implementation} reports on the performance of an Open-Source C++ implementation of the algorithm.

\section{Preliminaries}
\label{prelims}

Throughout this paper we use standard graph theory terminology as used in the textbook by Diestel \cite{diestel:graph}.  All graphs discussed here are simple and finite.  For a graph $G$, $V(G)$ and $E(G)$ denote the vertex and edge sets of $G$, respectively.  We use the terms \emph{vertex} and \emph{node} interchangeably, though we typically refer to the vertices of some primary graph $G$ of interest and refer to the nodes of some auxiliary graph (such as a spanning tree) related to $G$.  We say that a subgraph $G'$ of a graph $G$ \emph{spans} a set $S\subseteq V(G)$ if $S\subseteq V(G')$.

\paragraph*{Quotient Graphs.}

Given a graph $G$ and a partition $\mathcal{P}$ of $V(G)$, the \emph{quotient graph} $G/\mathcal{P}$ is the graph with vertex set $V(G/\mathcal{P}):=\mathcal{P}$ and in which two nodes $X,Y\in V(G/\mathcal{P})$ are adjacent if $G$ contains at least one edge $xy$ with $x\in X$ and $y\in Y$.

\paragraph*{Embeddings, Planar Graphs, and (Near-)Triangulations.}

An \emph{embedding} $\psi$ of a graph $G$ associates each vertex $v$ of $G$ with a point $\psi(v)\in \R^2$ and each edge $vw$ of $G$ with a simple open curve $\psi(vw):(0,1)\to\R^2$ whose endpoints\footnote{The \emph{endpoints} of an open curve $\psi:(0,1)\to\R^2$ are the two points $\lim_{\epsilon\downarrow 0} \psi(\epsilon)$ and $\lim_{\epsilon\downarrow 0}\psi(1-\epsilon)$.} are $\psi(v)$ and $\psi(w)$.  We do not distinguish between such a curve $\psi(vw)$ and the point set $\{\psi(vw)(t):0<t<1\}$.
We let $\psi(V(G)):=\{\psi(v):v\in V(G)\}$, $\psi(E(G)):=\bigcup_{vw\in E(G)} \psi(vw)$, and $\psi(G):=\psi(V(G))\cup\psi(E(G))$.  An embedding $\psi$ of $G$ is \emph{plane} if $\psi(vw)\cap\psi(V(G))=\emptyset$ and $\psi(vw)\cap\psi(xy)=\emptyset$ for each distinct pair of edges $vw,xy\in E(G)$.  A graph $G$ is \emph{planar} if it has a plane embedding. A \emph{triangulation} is an edge-maximal planar graph.

If $\psi$ is a plane embedding of a planar graph $G$, then we call the pair $(G,\psi)$ an \emph{embedded graph} and we will not distinguish between a vertex $v$ of $G$ and the point $\psi(v)$ or between an edge $vw$ of $G$ and the curve $\psi(vw)$.  Similarly, we will not distinguish between $G$ and the  point set $\psi(G)$.  Any cycle in an embedded graph defines a Jordan curve. For such a cycle $C$, $\R^2\setminus C$ has two components, one bounded and the other unbounded. We will refer to the bounded component as the \emph{interior} of $C$ and the unbounded component as the \emph{exterior} of $C$.  If $G$ is an embedded triangulation and $C$ is an arbitrary cycle in $G$, then the subgraph of $G$ consisting of all edges and vertices of $G$ contained in the closure of the interior of $C$ is called a \emph{near-triangulation}.

Each component of $\R^2\setminus G$ is a \emph{face} of $G$ and we let $F(G)$ denote the set of faces of $G$.  If $G$ is $2$-connected then, for any face $f\in F(G)$, the set of vertices and edges of $G$ contained in the boundary of $f$ forms a cycle.  We may therefore treat a face $f$ of a $2$-connected graph $G$ as a component of $\R^2\setminus G$ or as the cycle of $G$ on the boundary of $f$, relying on context to distinguish between the two usages.  Note that every embedded graph contains exactly one face---the \emph{outer face}---that is unbounded.

\paragraph*{Duals and Cotrees.}

The \emph{dual} $G^\star$ of an embedded graph $G$ is the graph with vertex set $V(G^\star):=F(G)$ and edge set $E(G^{\star}):=\{fg\in \binom{F(G)}{2}:E(f)\cap E(g)\neq\emptyset\}$.\footnote{For a set $S$, $\binom{S}{2}$ denotes the $\binom{|S|}{2}$-element set $\binom{S}{2}:=\{\{x,y\}:x,y\in S, x\neq y\}$.} If $T$ is a spanning tree of $G$ then the \emph{cotree} $\overline{T}$ of $(G,T)$ is the graph with vertex set $V(\overline{T}) := V(G^\star)$ and edge set $E(\overline{T}) := \{ab \in E(G^\star) : E(a) \cap E(b) \setminus E(T) \neq \emptyset \}$. It is well known that, if $G$ is connected and $T$ is a spanning tree of $G$ then $\overline{T}$ is a spanning tree of $G^\star$ \cite[Chapter~4, Exercise~42]{diestel:graph}.

For our purposes, a \emph{binary tree} is a rooted tree of maximum degree $3$ whose root has degree at most $2$ and in which each child $v$ of a node $u$ is either the unique \emph{left child} or the unique \emph{right child} of $u$.  If $G$ is a triangulation  and we root $\overline{T}$ at any face $f_0\in F(G)$ that contains an edge of $T$, then $\overline{T}$ is a binary tree, with the classification of left and right children determined by the embedding of $G$.\footnote{There is a small ambiguity here when $T$ contains two edges of $f_0$, in which case the unique child of $f_0$ in $\overline{T}$ can be treated as the left or right child of $f_0$.}

\paragraph*{Paths and Distances.}

A \emph{path} in $G$ is a (possibly empty) sequence of distinct vertices $v_0,\ldots,v_r$ with the property that $v_{i-1}v_i\in E(G)$, for each $i\in\{1,\ldots,r\}$.  The \emph{endpoints} of a path $v_0,\ldots,v_r$ are the vertices $v_0$ and $v_r$.
The \emph{length} of a non-empty path $v_0,\ldots,v_r$ is the number, $r$, of edges in the path.

\paragraph*{Trees, Depth, Ancestors, and Descendants.}

Let $T$ be a tree rooted at a vertex $v_0\in V(T)$.  For each $w\in V(T)$, $P_T(w)$ denotes the path in $T$ from $w$ to $v_0$.  For each $w_0\in V(T)$, any prefix $w_0,\ldots,w_r$ of $P_T(w_0)$ is called an \emph{upward path} in $T$; $w_0$ is the \emph{lower endpoint} of this path and $w_r$ is the \emph{upper endpoint}.  The \emph{$T$-depth} of a node $w\in V(T)$ is the length of the path $P_T(w)$. The second node in $P_T(v)$ (if any) is the \emph{$T$-parent} of $v$.  A vertex $a\in V(T)$ is a \emph{$T$-ancestor} of $w\in V(T)$ if $a\in V(P_T(w))$. If $a$ is a $T$-ancestor of $w$ then $w$ is a \emph{$T$-descendant} of $a$.

\paragraph*{Lowest Common Ancestors.}

For any two vertices $v,w\in V(T)$, the \emph{lowest common ancestor} $\lca_T(v,w)$ of $v$ and $w$ is the node $a$ in $P_T(v)\cap P_T(w)$ having maximum $T$-depth.  The \emph{lowest common ancestor problem} is a well-studied data structuring problem that asks to preprocess a given $n$-vertex rooted tree so that one can quickly return $\lca_T(v,w)$ for any two nodes $v,w\in V(T)$. A number of optimal solutions to this problem exist that, after $O(n)$ time preprocessing using $O(n)$ space, can answer queries in $O(1)$ time \cite{berman.vishkin:recursive,shieber.vishkin:on,harel.tarjan:fast,alstrup.gavoille.ea:nearest,bender.farach-colton:lca,fischer.heun:theoretical}.  The most recent work in this area includes simple and practical data structures that achieve this optimal performance \cite{alstrup.gavoille.ea:nearest,bender.farach-colton:lca,fischer.heun:theoretical}.

\paragraph*{Reconstructing Binary Tree Models.}

Let $T$ be a binary tree and $S \subseteq V(T)$. An \emph{upward path} $v_0,\ldots,v_r$ in a binary tree $T$ is \emph{S-non-branching} if $v_{i}$ has degree 2 and $v_i \notin S$ for each $i\in\{1,\ldots,r-1\}$. For any binary tree $T$ and set $S \subseteq V(T)$, the \emph{model} $T'$ of $T$ with respect to $S$ is the binary tree obtained by replacing each maximal $S$-non-branching path $v_0,\ldots,v_r$ with the edge $v_0v_r$; if $v_{r-1}$ is the left (respectively, right) child of $v_r$ then $v_0$ becomes the left (respectively, right) child of $v_r$. See \cref{trees}.

\begin{figure}[htbp]
  \begin{center}
    \begin{tabular}{ccc}
      \includegraphics[scale=0.75,page=1]{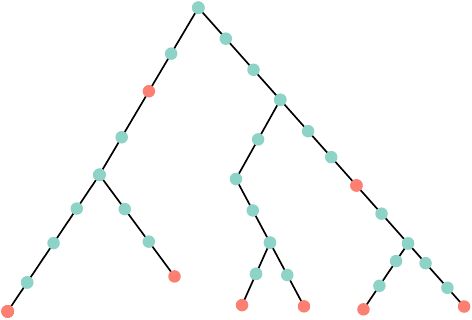} &
      &
      \includegraphics[scale=0.75,page=2]{figs/model} \\
    \end{tabular}
  \end{center}
  \caption{A binary tree $T$ with set $S \subseteq V(T)$ depicted in red and the model of $T$ with respect to $S$.}
  \label{trees}
\end{figure}

The following lemma shows that, given a constant time lowest common ancestor data structure for a rooted tree $T$ and a set $S$ of $d$ vertices in $T$, we can reconstruct the model $T_0'$ of a minimal subtree $T_0$ of $T$ that spans $S$ in time that depends only on $d$, and not on the size of $T$ or $T_0$.  In our algorithms we will use values of $d$ ranging from $1$ to $5$.

\begin{lem}\label{reconstruction}
  Let $T$ be a binary tree, let $S=\{x_1,\ldots,x_d\} \subseteq V(T)$, and let $T_0$ be the minimal subtree of $T$ that spans $S$. Then there exists an algorithm that, given an $O(1)$-query time lowest common ancestor data structure for $T$,  computes the model $T_0'$ of $T_0$ with respect to $S$ in $O(d^2)$ time.
\end{lem}

\begin{proof}
  The proof is by induction on $|S|$.  The base case $|S|=1$ is trivial, since then $T_0'=T_0$ is the tree with one node, which is the unique element in $S$.

  If $|S|\ge 2$, then the first step is to determine the root $r$ of $T_0$, which must also be the root of $T_0'$.  This is easily done by first setting $r:=x_1$ and then repeatedly setting $r:=\lca_{T}(r,x_i)$ for each $i\in\{2,\ldots,d\}$. This step takes $O(d)$ time.

  If $r$ has no left child in $T$, then we can immediately apply induction on $S\setminus\{r\}$ and make the right child of $r$ in $T_0'$ the root of the model obtained by induction.  The case in which $r$ has no right child can be handled similarly. If $r$ has both a left child $r_1$ and a right child $r_2$, then the next step is to partition $S\setminus\{r\}$ into a set $S_1$ of descendants of $r_1$ and a set $S_2$ of descendants of $r_2$.   For each $x\in S\setminus\{r\}$ there are only two possibilities for $\lca_{T}(r_1,x):$
  \begin{compactenum}
    \item If $\lca_{T}(r_1,x)=r_1$ then $x\in S_1$.
    \item If $\lca_{T}(r_1,x)=r$ then $x\in S_2$.
  \end{compactenum}
  Therefore, using $O(d)$ lowest common ancestor queries, we can determine the root $r$ of $T'$ and partition $S\setminus\{r\}$ into sets $S_1$ and $S_2$ that define the left and right subtrees of $r$.  We can now recurse on $S_1$ to obtain a tree with root $r_1'$ and recurse on $S_2$ to obtain a tree with root $r_2'$.  We make $r_1'$ the left child of $r$ and $r_2'$ the right child of $r$ to obtain the model $T_0'$ of $T_0$.  The running time of this  algorithm obeys the recurrence $T(d)\le O(d)+T(d_1) + T(d_2)$, where $d_1 + d_2 \leq d$ and $d_1, d_2 \leq d-1$.  This recurrence resolves to $T(d)\in O(d^2)$.
\end{proof}

\section{Tripod Partitions}
\label{tripod_decompositions}

Refer to \cref{tripod_decomposition_figure}. Let $G$ be an $n$-vertex triangulation and let $T$ be a spanning tree of $G$ rooted at a vertex $r\in V(G)$.   For a face $uvw$ of $G$, a \emph{$(G,T)$-tripod} $Y$ with \emph{crotch} $uvw$ is the vertex set of three disjoint (and each possibly empty) upward paths in $T$ (the \emph{legs} of $Y$) whose lower endpoints are $u$, $v$, and $w$.  
A \emph{$(G,T)$-tripod partition} is a partition of $V(G)$ into $(G,T)$-tripods.  \citet{dujmovic.joret.ea:planar} proved the following result:

\begin{figure}
  \begin{center}
    \begin{tabular}{cc}
      \includegraphics[page=1]{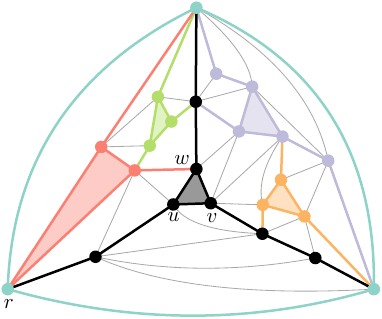} &
      \includegraphics[page=2]{figs/tripod-decomp}
    \end{tabular}
  \end{center}
  \caption{A $(G,T)$-tripod partition of a triangulation $G$ (and the underlying spanning tree $T$.)}
  \label{tripod_decomposition_figure}
\end{figure}

\begin{thm}\label{tripod_decomposition}
  Let $G$ be a triangulation and $T$ be a spanning tree of $G$.  Then there exists a $(G,T)$-tripod partition $\mathcal{Y}$ such that $G/\mathcal{Y}$ is planar and has treewidth at most $3$.
\end{thm}

A \emph{breadth-first spanning (BFS) tree} of a connected graph $G$ is a subgraph of $G$ that is a rooted tree and is rooted at a vertex $r$ of $G$ with the property that the path from $r$ to $v$ in $T$ is a shortest path from $r$ to $v$ in $G$, for each vertex $v\in V(G)$.  It is straightforward to verify that \cref{tripod_decomposition} implies \cref{meta}.\ref{three_tree} by first triangulating the given graph and then taking $T$ to be a BFS tree of the resulting triangulated graph \cite[Observation~35]{dujmovic.joret.ea:planar}. Briefly, the graph $H$ is defined as $H:=G/\mathcal{P}$; the length $r$ of the path $P:=p_0,\ldots,p_r$ is the height of the BFS tree $T$; and each tripod $Y$ in $\mathcal{P}$ maps to a subgraph of $u_Y\boxtimes P\boxtimes K_{3}$, where $u_Y$ denotes the vertex of $H$ corresponding to $Y$.  More specifically, a vertex $v\in Y$ that has depth $d$ in $T$, maps to the vertex $(u_Y,p_d,i_v)$ where $i_v\in\{1,2,3\}$ indicates the leg of $Y$ that contains $v$.

\subsection{Tripod Partitions from Face Orderings}
\label{orderings}

We now describe how a $(G,T)$-tripod partition can be obtained from a sequence of distinct faces of $G$.  Throughout this section (and for the remainder of the paper):
\begin{compactitem}
   \item $G$ is an embedded triangulation with outer face $f_0$ and
   \item $T$ is a spanning-tree of $G$ rooted at a vertex $v_0\in V(f_0)$.
\end{compactitem}
For any subgraph $f$ of $G$, we define $Y_T(f):=f\cup \bigcup_{v\in V(f)} P_T(v)$.\footnote{In all of our examples, the subgraph $f$ will always be a single edge or single face of $G$.}  In words, $Y_T(f)$ is the subgraph of $G$ that includes all the vertices and edges of $f$ and all the vertices and edges of each path from each vertex of $f$ to the root of $T$.

Let $\mathcal{F}:=f_0,\ldots,f_{r}$ be a sequence of distinct faces of $G$  whose first element is the outer face $f_0$. Let $G_{-1}$ denote the graph with no vertices and, for each $i\in\{0,\ldots,r\}$, define the graph $G_i:=\bigcup_{j=0}^i Y_T(f_j)$ and let $Y_i:=V(G_i)\setminus V(G_{i-1})$. 

Informally, we require that each of the legs of each tripod $Y_i$ have a \emph{foot}\footnote{Slightly more formally, the \emph{foot} of a leg $u_0,\ldots,u_\ell$ of a $(G,T)$-tripod is the parent of $u_\ell$ in $T$. An empty leg of a tripod has no foot, so a tripod may have zero, one, two, or three feet.} on a different vertex of $G_{i-1}$ and that the tripods $Y_1,\ldots,Y_r$ cover all the vertices and edges of $G$. Formally, we say that the sequence $\mathcal{F}$ is \emph{proper} if, for each $i\in\{1,\ldots,r\}$, and each distinct $v,w\in V(f_i)$, $V(P_T(v)\cap G_{i-1})\neq V(P_T(w)\cap G_{i-1})$.  The sequence $\mathcal{F}$ is \emph{complete} for $G$ if $G_r=G$.  Note that, if $\mathcal{F}$ is complete, then $\{Y_0,\ldots,Y_r\}$ is a tripod partition of $G$.

From the preceding definitions it follows that, if $\mathcal{F}$ is proper, then  $G_i$ is $2$-connected (since each face of $G_i$ is a cycle) for each $i\in\{0,\ldots,r\}$.  For any $i\in\{0,\ldots,r\}$, consider any face $f$ of $G_i$ that we now treat as a cycle in $G$. An easy proof by induction shows that, for any $j\in\{0,\ldots,i\}$, the induced graph $f[Y_j]$ is connected. %
We are interested in keeping the number of tripods in $Y_0,\ldots,Y_i$ that contribute to $V(f)$ as small as possible, which motivates our next definition.

The sequence $\mathcal{F}$ is \emph{good} if the resulting sequence of graphs $\mathcal{G}_\mathcal{F}:=G_0,\ldots,G_r$ and tripods $\mathcal{Y}_\mathcal{F}:=Y_0,\ldots,Y_r$ satisfy the following condition:  For each $i\in\{0,\ldots, r\}$ and each face $f$ of $G_i$,
\[
   |\{\ell\in\{0,\ldots,i\}: V(f)\cap Y_{\ell}\neq\emptyset\}|\le 3 \enspace .
\]
In words, each face $f$ of each graph $G_i$ has vertices from at most three tripods of $Y_0,\ldots,Y_i$ on its boundary.  Even more, the vertices of $f$ can be partitioned into at most three paths where the vertices of each path belong to a single tripod. \citet{dujmovic.joret.ea:planar} prove \cref{tripod_decomposition} by proving the next lemma.

\begin{lem}\label{face_trick}
  Let $G$ be a triangulation with a vertex $v_0$ on its outer face $f_0$ and let $T$ be a spanning tree of $G$ rooted at $v_0$.  Then there exists a sequence $\mathcal{F}:=f_0,\ldots,f_{r}$ of distinct faces of $G$ that is proper, good, and complete.
\end{lem}

\begin{rem}
  \cref{face_trick} is stated in terms of sequences only for convenience and could be rephrased in terms of partial orders. Indeed, consider the partial order $\prec$ defined as follows:  For each $i\in\{1,\ldots,r\}$ let $f_i'$ be the face of $G_{i-1}$ that contains $f_i$; then $f_\ell\prec f_i$ for each $\ell\in\{0,\ldots,i-1\}$ such that $V(f_i')\cap Y_\ell\neq\emptyset$.  It is straightforward to check that any linearization of this partial order will result in the same tripod partition $\mathcal{Y}_\mathcal{F}:=\{Y_0,\ldots,Y_{r}\}$.
\end{rem}

The key tool used by
\citet{dujmovic.joret.ea:planar} (and the earlier work by \citet{pilipczuk.siebertz:polynomial}) is Sperner's Lemma for $3$-coloured near-triangulations.  \citet{dujmovic.joret.ea:planar} prove \cref{face_trick} by giving a recursive algorithm that constructs the face sequence $\mathcal{F}$.  For a face $f$ of $G_i$, define the set $I_f:=\{\ell\in\{0,\ldots,i\}:V(f)\cap Y_\ell\neq\emptyset\}$.  They begin with the outer face $f_0$ of $G$.  To find the face $f_i$, $i>0$, they consider some face $f\not\in\{f_0,\ldots,f_{i-1}\}$ of $G_{i-1}$ and use Sperner's Lemma to show that there is an appropriate face $f_i$ of $G$ (called a \emph{Sperner triangle}) that is contained in $f$. In particular, $f_i$ is chosen so that the three upward paths in $Y_T(f_i)$ lead back to each of the (at most 3) tripods in $\{Y_j:j\in I_f\}$. See \cref{sperner}.

\begin{figure}
  \begin{center}
    \begin{tabular}{ccc}
      \includegraphics[page=1]{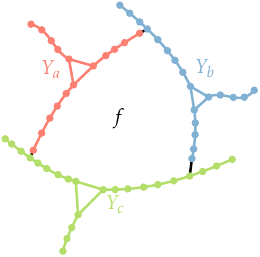} &
      \includegraphics[page=2]{figs/sperner} &
      \includegraphics[page=3]{figs/sperner} \\
      (a) & (b) & (c)
    \end{tabular}
  \end{center}
  \caption{Each face $f$ in $G_{i-1}$ is bounded by at most three tripods $Y_{a_f}$, $Y_{b_f}$, and $Y_{c_f}$ and the tripod $Y_i$ is chosen so that it connects each of these.}
  \label{sperner}
\end{figure}

This proof leads to a divide-and-conquer algorithm: After finding $f_i$, the algorithm recursively decomposes each of the near-triangulations that are bounded by the at most three new faces in $S_i:=F(G_i)\setminus F(G_{i-1})\setminus \{f_i\}$.  The Sperner triangle $f_i$ can easily be found in time proportional to the number of faces of $G$ in the interior of $f$.  However, because the resulting recursion is not necessarily balanced, a straightforward implementation of this yields an algorithm with $\Theta(n^2)$ worst-case running time.

\citet{morin:fast} later showed that, using an appropriate data structure for $T$, this approach can be implemented in such a way that the resulting algorithm runs in $O(n\log n)$ time.  Essentially, Morin's algorithm works by finding the Sperner triangle $f_i$ in time proportional to the minimum number of faces of $G$ contained in any of the faces in $S_i$.  In the next section, we will show that, by using a lowest common ancestor data structure for the cotree $\overline{T}$ along with \cref{reconstruction}, the Sperner triangle $f_i$ can be found in constant time, yielding an $O(n)$ time algorithm.

By now, our presentation of this material differs somewhat from that in \cite{dujmovic.joret.ea:planar,ueckerdt.wood.ea:improved}.  Therefore, we now pause to explain how \cref{face_trick} implies \cref{tripod_decomposition}. %
A graph $H$ is \emph{chordal} if there exists an ordering $v_1,\ldots,v_h$ of the vertices of $H$ such that, for each $i\in\{2,\ldots,h\}$, $N_G(v_i)\cap\{v_1,\ldots,v_{i-1}\}$ is a clique in the induced graph $G[\{v_1,\ldots,v_{i-1}\}]$.  Let $G$ be a triangulation, let $T$ be a spanning tree of $G$, let
$\mathcal{F}:=f_0,\ldots,f_r$ be the proper good face sequence guaranteed by \cref{face_trick}, and let $\mathcal{Y}_\mathcal{F}:=\{Y_0,\ldots,Y_r\}$ be the resulting tripod partition. We now show that there exists a chordal graph $H$ whose largest clique has size at most $4$ and that contains $G/\mathcal{Y}_\mathcal{F}$. We construct the graph $H$ so that for each $i\in\{0,\ldots,r\}$ and each face $f$ of $G_i$, $H$ contains a clique on $\{Y_j:j\in I_f\}$. To accomplish this, for each $i\in\{1,\ldots,r\}$ let $f$ be the face of $G_{i-1}$ that contains $f_i$ and form a clique on $\{Y_i\}\cup\{Y_j:j\in I_f\}$.  Inductively, the elements of $\{Y_j:j\in I_f\}$ already form a clique, so this operation is equivalent to attaching $Y_i$ to all the vertices of an existing clique of size at most $3$. Therefore, this results in a chordal graph $H$ whose largest clique has size at most $4$ and therefore $H$ has treewidth at most $3$ \cite{gavril:intersection}.

\section{\boldmath An \texorpdfstring{$O(n)$}{O(n)}-Time Algorithm}
\label{linear_time_algorithm}\label{algorithm}

Throughout the rest of this paper, we will assign colours to vertices of graphs.  We say that an edge is \emph{monochromatic} if both its endpoints have the same colour and \emph{bichromatic} if its endpoints have different colours.  A face is \emph{monochromatic}, \emph{bichromatic}, or \emph{trichromatic} if it contains vertices of one, two, or three distinct colours, respectively.  (The definition of good sequences will prevent the appearance of any faces with vertices of more than three different colours.)

Refer to \cref{baby_sperner_fig} for an illustration of the following (probably well-known) lemma, which is closely related to Sperner's Lemma:

\begin{figure}
  \begin{center}
    \includegraphics{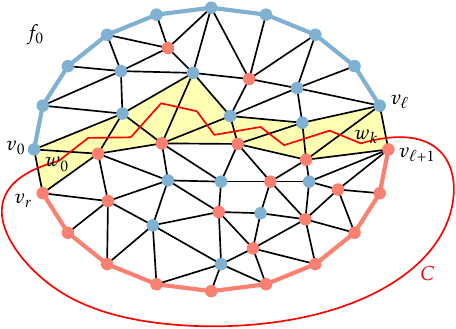}
  \end{center}
  \caption{\cref{baby_sperner}}
  \label{baby_sperner_fig}
\end{figure}

\begin{lem}\label{baby_sperner}
  Let $N$ be a near-triangulation with outer face $v_0,\ldots,v_r$ and colour each vertex of $N$ red or blue in such a way that $v_0,\ldots,v_\ell$ are coloured red for some $\ell\in\{0,\ldots,r-1\}$ and $v_{\ell+1},\ldots,v_r$ are coloured blue.  Then there exists a path $w_0,\ldots,w_k$ in the dual $N^\star$ of $N$ such that
  \begin{compactenum}
    \item $w_0$ is the inner face of $N$ with $v_0v_r$ on its boundary;
    \item $w_k$ is the inner face of $N$ with $v_{\ell}v_{\ell+1}$ on its boundary; and
    \item for each $i\in\{1,\ldots,k\}$, the single edge in $E(w_{i-1})\cap E(w_i)$ has an endpoint of each colour.
  \end{compactenum}
\end{lem}

\begin{proof}
  If $w_0=w_k$, the lemma is immediately true, so assume $w_0\neq w_k$. The outer face $f_0$ of $N$ has exactly two bichromatic edges $v_0v_r$ and $v_{\ell}v_{\ell+1}$ and any inner face of $N$ has either zero or two bichromatic edges.  Consider the subgraph $H$ of $N^\star$ obtained by removing each edge $fg\in E(N^\star)$ such that the edge in $E(f)\cap E(g)$ is monochromatic.  Every vertex in $H$ has degree $0$ or $2$, so each connected component of $H$ is either an isolated vertex or a cycle.  The face $f_0$ has degree $2$ so it is contained in a cycle $C$ of $H$.  The two neighbours of $f_0$ in $H$ are $w_0$ and $w_k$. Therefore $C$ contains a path $w_0,\ldots,w_k$ that satisfies the conditions of the lemma.
\end{proof}

The next lemma, which is the main new insight in this paper, allows us to use \cref{reconstruction} (which relies on a lowest common ancestor data structure) to find Sperner triangles in constant time.  In particular, in the critical case described in this lemma, the Sperner triangle can be found by constructing the model $\overline{T}_0'$ of the minimal subtree $\overline{T}_0$ of $\overline{T}$ that spans the set $S:=\{g_1,g_2,g_3\}$.  Specifically, the Sperner triangle $\tau$ is the unique node of $\overline{T}_0'$ such that each component of $\overline{T}_0'-\tau$ contains at most one of $g_1$, $g_2$ or $g_3$.

\begin{lem}\label{lca_sperner}
  Let $G$ be a triangulation with a vertex $v_0$ on its outer face $f_0$; let $T$ be a spanning tree of $G$ rooted at $v_0$; let $\overline{T}$ be the cotree of $(G,T)$ rooted at $f_0$; let $f_0,\ldots,f_{i-1}$ be a good proper sequence of faces of $G$ that yields a sequence $\mathcal{G}_\mathcal{F}:=G_0,\ldots,G_{i-1}$ of graphs and a sequence $\mathcal{Y}_\mathcal{F}:=Y_0,\ldots,Y_{i-1}$ of tripods; let $f\not\in \{f_0,\ldots,f_{i-1}\}$ be a face of $G_{i-1}$, and let $S$ be a set of three distinct faces of $G$ such that, for each $g\in S$,
  \begin{compactenum}[(i)]
      \item $g$ is contained in the interior of $f$;
      \item $g$ contains an edge $vw\in E(f)$ with $v\in Y_a$ and $w\in Y_b$ for some distinct $a,b\in I_f=\{\ell\in\{0,\ldots,i\}:V(f)\cap Y_\ell\neq\emptyset\}$.
  \end{compactenum}
  Let $\overline{T}_0$ be the minimal subtree of $\overline{T}$ that spans $S$.  If $f_i\in V(\overline{T}_0)$ is such that each component of $\overline{T}_0-f_i$ contains at most one element of $S$, then $f_0,\ldots,f_i$ is good.
\end{lem}

\begin{proof}
  Let $N$ be the near-triangulation consisting of all vertices and edges of $G$ contained in the closure of the interior of $f$.  
  Since $f_0,\ldots,f_{i-1}$ is good, $|I_f|\le 3$.  Since $|S|=3$, $|I_f|=3$.  For each $j\in I_f$, colour each vertex $v$ of $N$ with the colour $j$ if the first vertex of $P_{T}(v)$ in $V(f)$ is contained in $Y_j$.  

  $E(f)$ contains exactly three bichromatic edges, each of which is on the boundary of exactly one face in $S$.  Let $X$ be the subgraph of $N^\star$ that contains an edge $fg\in E(N^\star)$ if and only if $f$ and $g$ are inner faces of $N$ and the edge in $E(f)\cap E(g)$ is bichromatic.  We claim that $X$ is a subgraph of $\overline{T}$.  In order to show this, we need only argue that each edge $uv$ of $T$ in the interior of $f$ is monochromatic.  Consider any $uv\in E(N)\setminus E(f)$ where $u$ is the $T$-parent of $v$.  If $v\not\in V(f)$ then, by definition, $v$ has the same colour as $u$, so $uv$ is monochromatic. The case where $v\in V(f)$ and $u\not\in V(f)$ cannot occur since $v\in V(f)$ implies that $P_T(v)\subseteq G_{i-1}$, but $u\not\in V(G_{i-1})$.  Similarly, the case in which $u\in V(f)$ and $v\in V(f)$ cannot occur since this implies that $P_T(v)\subseteq G_{i-1}$, but $uv\not\in E(G_{i-1})$.

  Next we claim that all the elements of $S$ are in a single connected component of $X$.  To see this, let $\{a,b,c\}:=I_f$ and consider a pair $g_1,g_2\in S$ where (without loss of generality) $g_1$ contains a bichromatic edge of $f$ with colours $a$ and $b$, and $g_2$ contains a bichromatic edge of $f$ with colours $b$ and $c$.  By treating $a$ and $c$ as a single colour, we apply \cref{baby_sperner} to conclude that $g_1$ and $g_2$ are in the same component of $X$.

  Refer to \cref{lca_sperner_fig}(a).
  Therefore $X$ is a subgraph of $\overline{T}$ that has a component containing all the elements of $S$. Therefore $X$ contains $\overline{T}_0$.  By choice, $\overline{T}_0$ contains a path from $f_i$ to each $g\in S$ and these paths are pairwise disjoint except for their shared starting point $f_i$.   

  \begin{figure}
    \begin{center}
      \begin{tabular}{ccc}
        \includegraphics[page=1]{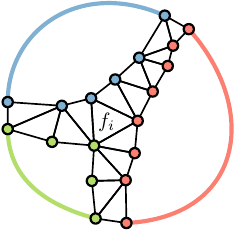} &
        \includegraphics[page=2]{figs/fresh} &
        \includegraphics[page=4]{figs/fresh} \\
        (a) & (b) & (c)
      \end{tabular}
    \end{center}
    \caption{The proof of \cref{lca_sperner}}
    \label{lca_sperner_fig}
  \end{figure}

  Refer to \cref{lca_sperner_fig}(b).
  Now, consider the embedded graph $X_0$ obtained as follows: For each $g\in V(\overline{T}_0)$, place a vertex on the center of each bichromatic edge of $g$ and, if $g$ is trichromatic, then place a vertex in the center of $g$.  Next,
   \begin{compactenum}
    \item add an edge joining the center of each trichromatic triangle to each of the centers of its bichromatic edges; and
    \item add an edge (embedded as a straight line segment) joining the centers of each pair of bichromatic edges that are on a common bichromatic face $g\in V(\overline{T}_0)$.
  \end{compactenum}
  The graph $X_0$ is a tree of maximum-degree $3$ that has three leaves.  (Each leaf in $X_0$ is the center of a bichromatic edge of $E(f)$). Therefore $X_0$ has exactly one degree-$3$ vertex.  We claim that this degree-$3$ vertex is in the interior of $f_i$.  Indeed, if this were not the case,  then $\overline{T}_0-f_i$ would necessarily have a component $Z$ such that $\bigcup_{g\in V(Z)} E(g)$ contains at least two bichromatic edges of $E(f)$.  But this would imply that $Z$ contains two elements of $S$, contradicting the definition of $f_i$.
  
  Refer to \cref{lca_sperner_fig}(c).
  Now treat $X_0$ as a point set and consider the point set $f'$ obtained by removing $X_0$ from the closure of $f$. Now $f'$ has $|I_f|=3$ connected components and each vertex of $f_i$ is in a different component.  Each of the components of $f'$ contains vertices of $Y_j$ for exactly one $j\in I_f$; call this component $f'_j$.  Since no edge of $T$ crosses $X_0$, the colour of each vertex $v$ in $f'_j$ is equal to $j$, for each $j\in I_f$.  Thus, each vertex of $f_i$ has a distinct colour.

  To see that $f_0,\ldots,f_i$ is good, first observe that we need only be concerned with the at most three faces in $F(G_i)\setminus F(G_{i-1})\setminus\{f_i\}$.  Let $g\in F(G_i)\setminus F(G_{i-1})\setminus\{f_i\}$.  Then $g$ has exactly two vertices $v$ and $w$ of $f_i$ on its boundary.  Without loss of generality $v$ and $w$ have colours $a$ and $b$ with $a\neq b$.  Then all vertices of $g$ are contained in $f'_a\cup f'_b$. Therefore, all vertices of $V(g)$ have colour $a$ or $b$.  Since $V(g)\subseteq V(f) \cup Y_i\subseteq Y_a\cup Y_b\cup Y_c \cup Y_i$, we have  $V(g)\subseteq Y_a\cup Y_b\cup Y_i$.  Therefore $|I_g|=|\{\ell\in\{0,\ldots,i\}:V(g)\cap Y_{\ell}\neq\emptyset\}| \le |\{a,b,i\}| = 3$. Since this is true for each $g\in F(G_i)\setminus F(G_{i-1})\setminus\{f_i\}$, $f_0,\ldots,f_i$ is good.
\end{proof}

We now pause to explain the representation of our triangulation. For concreteness, we describe our algorithms using a half-edge data structure, variants of which are described in \cite{baumgart.ea:vision, guibas.stolfi.ea:voronoi, kettner.ea:visualization, weiler.ea:structure}.  The important property of this structure is that the following operations are supported in constant time: For every directed half-edge $vw$, we can find the reversal $wv$ of $vw$, the next directed half-edge $vw'$ that appears in counterclockwise order around $v$, and the face to the left of $vw$. For every face $f$, we can find the three neighbouring faces of $f$, the three incident vertices of $f$, and the three half-edges having $f$ to their left.

\begin{thm}\label{three_tree_algorithm}
  There exists an $O(n)$ time algorithm that, given any $n$-vertex triangulation $G$ and any rooted spanning tree $T$ of $G$, produces a $(G,T)$-tripod partition $\mathcal{Y}$ such that $\tw(G/\mathcal{Y})\le 3$.
\end{thm}

\begin{proof}
  Let $v_0$ be the root of $T$ and let $f_0$ be a face of $G$ incident to $v_0$ that contains an edge of $T$ incident to $v_0$.  In a preprocessing step, we compute the cotree $\overline{T}$ of $(G,T)$ and construct a lowest common ancestor data structure for $\overline{T}$ in $O(n)$ time that allows us to compute $\lca_{\overline{T}}(f,g)$ for any two faces $f,g\in F(G)$ in $O(1)$ time.

  After this preprocessing, we construct the good sequence $f_0,\ldots,f_{r}$ recursively. Conceptually, during any recursive invocation, the input is a near-triangulation $N$ bounded by a cycle $C$ in $G$ whose vertices belong to at most three tripods computed in previous steps.  Each vertex of $G$ starts initially \emph{uncoloured} and we \emph{colour} a vertex $v$ in a tripod $Y_i$ with the colour $i$ when tripod $Y_i$ is created.  Likewise, each face of $G$ starts initially uncoloured and we colour a face once we have added it to the good sequence.  Initially, only the three vertices of $f_0$ are coloured and each is assigned the colour $0$ of the tripod $Y_0$.  There are three types of recursive instances.  The initial recursive call is a monochromatic instance, described below.  

  In the following, we first describe the precise input to each type of recursive instance and the resulting Sperner triangle $f_i$ arising from this type of instance.  Then, for each type of instance, we describe how to construct the recursive subproblems (corresponding to faces of $G_i$ that are not present in $G_{i-1}$) that are generated by the instance. The precise input to a recursive invocation is defined as follows:  
\paragraph{Trichromatic instances.}
If $C$ intersects three tripods, then the input consists of the three bichromatic half-edges $e_1$, $e_2$, and $e_3$ of $C$ (arranged in counterclockwise order around $C$) that have $N$ to their left. 
For each $i\in\{1,2,3\}$, let $g_i$ be the inner face of $N$ to the left of $e_i$.

If the faces $g_1$, $g_2$, and $g_3$ are all distinct, then 
\cref{lca_sperner} characterizes the face $f_i$ in terms of the minimum subtree $\overline{T}_0$ of $\overline{T}$ that contains $g_1$, $g_2$, and $g_3$.  Indeed, $f_i$ is either the unique degree-$3$ node of $\overline{T}_0$ (if $g_1$, $g_2$, and $g_3$ are all leaves of $\overline{T}_0$) or $f_i$ is the unique node among $g_1$, $g_2$, and $g_3$ that has degree $2$.  By \cref{reconstruction} we can construct the model $\overline{T}_0'$ of $\overline{T}_0$ in constant time and find the node $f_i$.  If the faces $g_1$, $g_2$ and $g_3$ are not all distinct because, for instance $g_1=g_2$, then the face $g_1$ has at least two edges ($e_1$ and $e_2$) that belong to $C$, and the three vertices of $g_1$ belong to three different tripods. In this case we choose $f_i=g_1$.

\paragraph{Bichromatic instances.}
If $C$ intersects two tripods, then the input consists of the two half-edges $e_1$ and $e_2$ of $C$. By extension, $g_1$ and $g_2$ are the two inner faces of $N$ to the left of these half-edges.  In this case, we let $f_i:=g_1$ or $f_i=g_2$, either choice satisfies our requirements.
\paragraph{Monochromatic instances.}
If $C$ intersects only one tripod, then the input consists of some half-edge $e_1$ of $C$. By extension, $g_1$ is the inner face of $N$ to the left of this half-edge.  In this case $f_i:=g_1$ satisfies our requirements.

\paragraph{Constructing tripods and identifying subproblems.}
Once we have found $f_i$, we can compute the tripod $Y_i$ and mark its vertices by following the path in $T$ from each vertex of $f_i$ to its nearest marked ancestor in $T$.  This takes $O(1+|Y_i|)$ time.  Once we have done this, we have also found the at most three bichromatic edges of $G_i$ that are needed to perform the at most three recursive invocations on the near-triangulations whose outer faces coincide with each of the new faces in $F(G_i)\setminus F(G_{i-1})\setminus\{f_i\}$. To identify these recursive invocations in constant time, we need to efficiently:
\begin{compactenum}
 \item recognize some candidate near-triangulation as being non-empty, and
 \item determine the bichromatic half-edges needed to perform said recursive invocation.
\end{compactenum}

We begin with the first issue.  Let $uv$ be a half-edge of $f_i$ with $f_i$ to the left of $uv$. 
There are three situations that may cause the subproblem to the right of $uv$ to be empty:
\begin{compactenum}
 \item $uv$ is an edge of $C$.  To detect this we distinguish between two cases.  If $uv$ is bichromatic then $uv$ belongs to $C$ if and only if $uv\in\{e_1,e_2,e_3\}$. If $uv$ is monochromatic then $uv$ belongs to $C$ if and only if $u$ is the $T$-parent of $v$ (or vice-versa) or $uv$ is an edge of $f_j$ for some $j<i$. The former situation can easily be checked directly. The latter situation can be identified by checking the colour (if any) of the face to the right of the half-edge $uv$.
 
 \item $u\in Y_i$ and $v$ is the $T$-parent of $u$.
 
 \item $v\in Y_i$ and $u$ is the $T$-parent of $v$.

\end{compactenum}
Note that points 1 and 2 are mutually exclusive, since if $uv$ is an edge of $C$ then $u\not\in Y_i$. Similarly, points 1 and 3 are also mutually exclusive.  The precise subproblem identification process is broken down as follows:

\paragraph{Recursive subproblems from monochromatic instances.} %
Monochromatic instances induce at most two recursive invocations. We must use points 1, 2, and 3 above to identify subproblems. Let $uvw$ be the vertices of $f_i$, with $vw = e_1$. Then, there exists a subproblem to the right of $uv$, and a subproblem to the right of $wu$.  We describe how to treat the subproblem to the right of $uv$ using points $1$ and $2$.  There are two cases to consider: 
\begin{enumerate}[nosep,nolistsep]
  \item If $u$ is a vertex of $C$, then this subproblem is empty if and only if $uv$ is an edge of $C$ (point 1). If this subproblem is non-empty, then it is a monochromatic subproblem whose input is the half-edge $e_1'=vu$.

  \item If $u$ is not a vertex of $C$, then $u\in Y_i$ and the subproblem to the right of $uv$ is empty if and only if $v$ is the $T$-parent of $u$ (point 2).  If this subproblem is non-empty, then it is a bichromatic subproblem whose input is the half-edge $e_1'=vu$ and the half-edge 
  $e_2'=u'u''$, where $u'$ is the last vertex in the leg of $Y_i$ that contains $u$ and $u''$ is the $T$-parent of $u'$.
\end{enumerate}
The subproblem to the right of $wu$ is treated similarly using points $1$ and $3$.

\paragraph{Recursive subproblems from bichromatic instances.} %
Similar to monochromatic instances, bichromatic instances induce at most two recursive invocations. We must use points 1, 2, and 3 above to identify subproblems. Let $uvw$ be the vertices of $f_i$, with $vw = e_1$. Then, there exists a subproblem to the right of $uv$, and a subproblem to the right of $wu$.  We describe how to treat the subproblem to the right of $uv$ using points $1$ and $2$.  There are two cases to consider:

\begin{enumerate}[nosep,nolistsep]
  \item If $u$ is a vertex of $C$, then this subproblem is empty if and only if $uv$ is an edge of $C$ (point 1). If this subproblem is non-empty, then it is either a monochromatic subproblem or a bichromatic subproblem.
  
  \begin{enumerate}[nosep,nolistsep]
  \item A non-empty subproblem is monochromatic if $u$ and $v$ are coloured identically (that is, $uv$ is monochromatic). The input to a monochromatic subproblem is the half-edge $e_1'=vu$.
  
  \item A non-empty subproblem is bichromatic if $u$ and $w$ are coloured identically (that is, $uv$ is bichromatic). The input to a bichromatic subproblem is the half-edge $e_1'=vu$ and the half-edge $e_2'=e_2$.
\end{enumerate}

  \item If $u$ is not a vertex of $C$, then $u\in Y_i$.  Then the subproblem to the right of $uv$ is empty if and only if $v$ is the $T$-parent of $u$ (point 2).  Let $u'$ denote the last vertex in the leg of $Y_i$ that contains $u$ and let $u''$ be the $T$-parent of $u'$.  If the subproblem to the right of $uv$ is non-empty, then it is either a bichromatic subproblem or a trichromatic subproblem.
  
\begin{enumerate}[nosep,nolistsep]
  \item A non-empty subproblem is bichromatic if $v$ and $u''$ are coloured identically. The input to a bichromatic subproblem is the half-edge $e_1'=vu$ and the half-edge $e_2'=u'u''$.
  
  \item A non-empty subproblem is trichromatic if $w$ and $u''$ are coloured identically. The input to a trichromatic subproblem is the half-edge $e_1'=vu$, the half-edge $e_2'=u'u''$, and the half-edge $e_3'=e_2$.
\end{enumerate}
\end{enumerate}
The subproblem to the right of $wu$ is treated similarly using points $1$ and $3$.

\paragraph{Recursive subproblems from trichromatic instances.} %
Trichromatic instances induce at most three recursive invocations. We must use point 1 above to identify subproblems. Let $uvw$ be the vertices of $f_i$. Then, there exists a subproblem to the right of $uv$, a subproblem to the right of $vw$, and a subproblem to the right of $wu$.  Without loss of generality, suppose that the subproblem to the right of $uv$ contains $e_1$, the subproblem to the right of $vw$ contains $e_2$, and the subproblem to the right of $wu$ contains $e_3$.

We describe how to treat the subproblem to the right of $uv$. There are four cases to consider:
\begin{enumerate}[nosep,nolistsep]
  \item If $u$ is a vertex of $C$ and $v$ is a vertex of $C$, then this subproblem is empty if and only if $uv$ is an edge of $C$ (point 1). If this subproblem is non-empty, then it is a bichromatic subproblem whose input is the half-edge $e_1'=vu$ and the half-edge $e_2'=e_1$.

  \item If $u$ is a vertex of $C$ and $v$ is not a vertex of $C$, then this subproblem is always a non-empty trichromatic subproblem. Let $v'$ denote the last vertex in the leg of $Y_i$ that contains $v$ and let $v''$ be the $T$-parent of $v'$. Then the input to this subproblem is the half-edge $e_1'=v''v'$, the half-edge $e_2'=vu$, and the half-edge $e_3'=e_1$.

  \item If $u$ is not a vertex of $C$ and $v$ is a vertex of $C$, then this subproblem is always a non-empty trichromatic subproblem. Let $u'$ denote the last vertex in the leg of $Y_i$ that contains $u$ and let $u''$ be the $T$-parent of $u'$. Then the input to this subproblem is the half-edge $e_1'=u'u''$, the half-edge $e_2'=e_1$, and the half-edge $e_3'=vu$.

  \item If $u$ is not a vertex of $C$ and $v$ is not a vertex of $C$, then this subproblem is always a non-empty trichromatic subproblem.  Let $u'$ denote the last vertex in the leg of $Y_i$ that contains $u$, let $u''$ be the $T$-parent of $u'$, let $v'$ denote the last vertex in the leg of $Y_i$ that contains $v$, and let $v''$ be the $T$-parent of $v'$.  Then the input to this subproblem is the half-edge $e_1'=v''v'$, the half-edge $e_2'=u'u''$, and the half-edge $e_3'=e_1$.

\end{enumerate}
The subproblem to the right of $vw$ and the subproblem to the right of $wu$ are treated similarly using point 1.

As stated above, the initial recursive call is a monochromatic instance. Its input is any of the three half-edges adjacent to the outer face $f_0$ that have $G \setminus f_0$ to their left.  Each recursive invocation adds a new face $f_i$ to the good face sequence $f_0,\ldots,f_{r}$ and takes $O(1+|Y_i|)$ time.  Since $Y_0,\ldots,Y_{r}$ is a partition of $V(G)$, the running time of this algorithm is therefore $\sum_{i=0}^{r} O(1+|Y_i|) = O(n)$.
\end{proof}

\section{Variations}
\label{variants}

In this section we show that there are $O(n)$ time algorithms for computing the partitions in \cref{meta}.\ref{four_tree} and \cref{meta}.\ref{six_tree}.  In the same way that \cref{meta}.\ref{three_tree} follows from the tripod partition of \cref{tripod_decomposition}, \cref{meta}.\ref{four_tree} follows from a bipod partition given by \cref{bipod_decomposition_algorithm}, and \cref{meta}.\ref{six_tree} follows from a monopod partition given by \cref{monopod_decomposition_algorithm}.

\subsection{Bipod Partitions}

We begin with the partition in \cref{meta}.\ref{four_tree}, which was communicated to us by Vida~Dujmović, and has not appeared before.  This partition is obtained by selecting a proper sequence $\mathcal{E}:=e_0,\ldots,e_r$ of distinct edges of $G$, which define a sequence of graphs $\mathcal{G}_\mathcal{E}:=G_{0},\ldots,G_r$ where $G_i:=\bigcup_{j=0}^i Y_T(e_j)$ and a sequence of \emph{bipods} $\mathcal{I}_\mathcal{E}:=\Lambda_0,\ldots,\Lambda_r$ where $\Lambda_i=V(G_i)\setminus V(G_{i-1})$.  We call $\mathcal{E}$ \emph{good} if, for each $i\in\{0,\ldots,r\}$ and each face $f\in F(G_i)$, $V(f)$ has a non-empty intersection with at most $4$ bipods in $\Lambda_0,\ldots,\Lambda_i$.

Exactly the same argument used in \cref{orderings} to show that $G/\mathcal{Y}_\mathcal{F}$ is contained in a chordal graph of maximum clique size $4$ also shows that if $\mathcal{E}$ is a good edge sequence that produces a bipod partition $\mathcal{I}_\mathcal{E}$ of $V(G)$, then $G/\mathcal{I}_\mathcal{E}$ is contained in a chordal graph of maximum clique size $5$, so $G/\mathcal{I}_\mathcal{E}$ has treewidth at most $4$.

We now explain why a good edge sequence $e_0,\ldots,e_r$ exists.\footnote{The existence of this edge sequence is more easily proven using Sperner's Lemma, but we want a proof that lends itself to a linear time algorithm.}  As before, we set $f_0$ to be any face of $G$ such that $E(f_0)$ contains an edge of $T$ incident to the root $v_0$ of $T$.  The edge $e_0$ is any edge of $E(f_0)\setminus E(T)$.  Next we take special care to ensure that $G_i$ is biconnected for $i\ge 1$.  In particular, if $G_0$ contains only two edges of $f_0$, then we take $e_1$ to be the edge of $f_0$ that does not appear in $G_0$.  Otherwise, we choose $e_1$ using the general strategy for choosing $e_i$, described next.

Refer to \cref{e_i}.  Now we may assume that $G_{i-1}$ is biconnected. To choose the edge $e_i$, we consider any face $f\in F(G_{i-1})\setminus F(G)$. Inductively, $V(f)$ contains vertices from at most four bipods in $\Lambda_0,\ldots,\Lambda_{i-1}$. Let $I_f:=\{j\in\{0,\ldots,i-1\}:\Lambda_j\cap V(f)\neq\emptyset\}$. If $|I_f| < 4$ then we can select $e_i$ to be any edge in the interior $f$. Therefore, we focus on the case $|I_f| = 4$. As before we colour vertices in the near-triangulation $N$ using colours in the set $I_f$; we let $S$ be the set of inner faces in $N$ that contain a bichromatic edge in $E(f)$; and let $\overline{T}_0$ be the minimal subtree of $\overline{T}$ that spans $S$.  The same argument in the proof of \cref{lca_sperner} shows that every node of $\overline{T}_0$ is contained in $f$.

\cref{director}, below, shows that $\overline{T}_0$ contains an edge $xy$ such that each component of $\overline{T}_0-xy$ contains at most two elements of $S$.  It is straightforward to verify that, if we choose $e_i$ to be the edge in $E(x)\cap E(y)$ then we obtain a graph $G_i$ in which each of the two new faces containing vertices from $\Lambda_i$ contains vertices from at most three bipods in $\{\Lambda_j:j\in I_f\}$, as required.

\begin{figure}[htbp]
  \begin{center}
    \begin{tabular}{cc}
       \includegraphics[page=1]{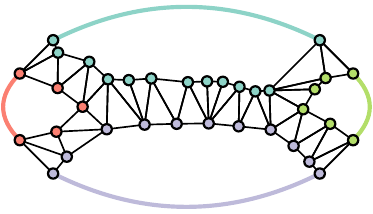} &
       \includegraphics[page=3]{figs/two_up}
     \end{tabular}
  \end{center}
  \caption{Choosing the next edge $e_i$ in a good edge sequence.}
  \label{e_i}
\end{figure}

\begin{clm}\label{director}
  $\overline{T}_0$ contains an edge $xy$ such that each component of $\overline{T}_0-xy$ contains at most two nodes of $S$.
\end{clm}

\begin{proof}
  Direct each edge $xy$ of $\overline{T}_0$ in the direction $\overrightarrow{xy}$ if the component of $\overline{T}_0-xy$ that contains $y$ contains three or more nodes of $S$.  It is sufficient to show that this process leaves some edge $xy$ of $\overline{T}_0$ undirected.  Assume for the sake of contradiction that every edge of $\overline{T}_0$ is directed.  Then some node $x$ of $\overline{T}_0$ has only incoming edges. Certainly $x$ does not have degree $1$ in $\overline{T}_0$.

  If $x$ has degree $2$ in $\overline{T}_0$ then $\overline{T}_0$ contains two subtrees $T_1$ and $T_2$ that have only the node $x$ in common and such that $|V(T_1)\cap S|\ge 3$ and $|V(T_2)\cap S|\ge 3$, which implies that $|S|\ge 3+3-1 > 4$, a contradiction.

  Suppose therefore that $x$ has degree $3$ in $\overline{T}_0$.  Each face in $S$ contains an edge in $E(f)$, so each face in $S$ has degree at most $2$ in $\overline{T}_0$.  Therefore $x\not\in S$.  Therefore $\overline{T}_0-x$ contains three components $T_1, T_2, T_3$ such that each pair of components contains at least $3$ elements of $S$.  But this implies that $|S|\ge (3\times 3)/2>4$, a contradiction.
\end{proof}

Algorithmically, using \cref{reconstruction}, we can construct the model $\overline{T}_0'$ of $\overline{T}_0$ in constant time given the set $S$.  The model $\overline{T}_0'$ will also contain an edge $\alpha\beta$ such that each component of $\overline{T}_0'-\alpha\beta$ contains at most two nodes of $S$. We claim that $E(\alpha)$ contains an edge that makes a suitable choice for $e_i$, and this edge can be found in constant time.  Indeed, the edge $\alpha\beta$ in  $\overline{T}_0'$ corresponds to a path $\alpha,x_1,\ldots,x_k,\beta$ in $\overline{T}_0$ and the unique edge in $E(\alpha)\cap E(x_1)$ is a suitable choice for $e_i$.  

The rest of the details of the algorithm are similar to those given in the proof of \cref{three_tree_algorithm}:  Each subproblem is a near-triangulation $N$ bounded by a cycle $C$ and the input that defines the subproblem consists of the (at most four) bichromatic half-edges of $C$.
(In the degenerate case where $C$ has no bichromatic edges, the input is any half-edge of $C$.)
Similar issues arise when determining which subproblems are non-empty and determining the half-edges used as input for each non-empty subproblem.  The interested reader can find these details in the code \cite{morin:fatripar,gouin:planar}.

\begin{thm}\label{bipod_decomposition_algorithm}
  There exists an $O(n)$ time algorithm that, given any $n$-vertex triangulation $G$ and any rooted spanning tree $T$ of $G$, produces a $(G,T)$-bipod partition $\mathcal{I}$ such that $\tw(G/\mathcal{I})\le 4$.
\end{thm}

\subsection{Monopod Partitions}

Finally we consider the partition described in \cref{meta}.\ref{six_tree}.   This partition is obtained from a tripod partition $\mathcal{Y}:=Y_0,\ldots,Y_r$, obtained by a sequence $\mathcal{F}:=f_0,\ldots,f_r$ of faces of $G$ in the same manner described in \cref{orderings}.  However in this setting, the sequence $f_0,\ldots,f_r$ is \emph{good} if, for each $i\in\{0,\ldots,r\}$ and each face $f$ of $G_i:=\bigcup_{j=0}^i Y_T(f_j)$, $V(f)$ contains vertices from at most $5$ legs of tripods in $Y_{0},\ldots,Y_i$.  Under these conditions, \citet{ueckerdt.wood.ea:improved} are able to show that the \emph{monopod partition} $\mathcal{I}$ obtained by splitting each tripod $Y_i$ into three upward paths yields a quotient graph $G/\mathcal{I}$ of treewidth at most $6$.

As before we focus on the extreme case when $V(f)$ contains vertices from exactly 5 legs of monopods. Refer to \cref{uwy}. Following the same strategy used for the previous two partitions, the set $S$ in this case has size at most $5$ and the face $f_i$ corresponds to a node of $\overline{T}_0$ such that each component of $\overline{T}_0-f_i$ contains at most $2$ nodes of $S$.  (This is always possible because $\lfloor 5/2\rfloor=2$.)  Again, a suitable choice for $f_i$ can be found in the model $\overline{T}_0'$ of $\overline{T}_0$ in constant time.

\begin{figure}[htbp]
  \begin{center}
    \begin{tabular}{cc}
       \includegraphics[page=1]{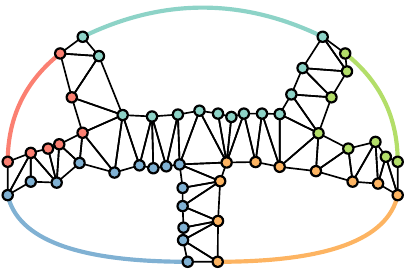} &
       \includegraphics[page=3]{figs/uwy2}
     \end{tabular}
  \end{center}
  \caption{The selection of a tripod by \citet{ueckerdt.wood.ea:improved}}
  \label{uwy}
\end{figure}

The details of determining which subproblems are non-empty and which half-edges are used as input for each subproblem are similar to the previous two partitions. Details can be found in the source code \cite{morin:fatripar,gouin:planar}.

\begin{thm}\label{monopod_decomposition_algorithm}
  There exists an $O(n)$ time algorithm that, given any $n$-vertex triangulation $G$ and any rooted spanning tree $T$ of $G$, produces a $(G,T)$-monopod partition $\mathcal{I}$ such that $\tw(G/\mathcal{I})\le 6$.
\end{thm}

\section{Implementation}
\label{implementation}

The algorithms described in this paper have been implemented in C++ \cite{morin:fatripar} and in C \cite{gouin:planar}.  The C++ implementation includes the algorithms for tripod, bipod, and monopod partitions. The C implementation implements only the tripod decomposition (which, at the time of writing, gives the tightest results in theoretical applications).

As discussed in the proof of \cref{three_tree_algorithm}, the implementation requires some care when deciding which subproblems are non-empty, but the resulting code is not overly long. The entire C++ implementation of all three algorithms, including the code for the lowest common ancestor data structure, has 924 lines of code.\footnote{A \emph{line of code} is any line that contains a semicolon.}  The resulting code is also quite efficient. \Cref{plots} shows the total running time of the implementation on $n$-vertex triangulations for $n$ ranging from $10,000$ to $1,990,000$ in increments of $10,000$.  Each measurement is the average of $20$ runs on $20$ different triangulations.

\begin{figure}
  \includegraphics{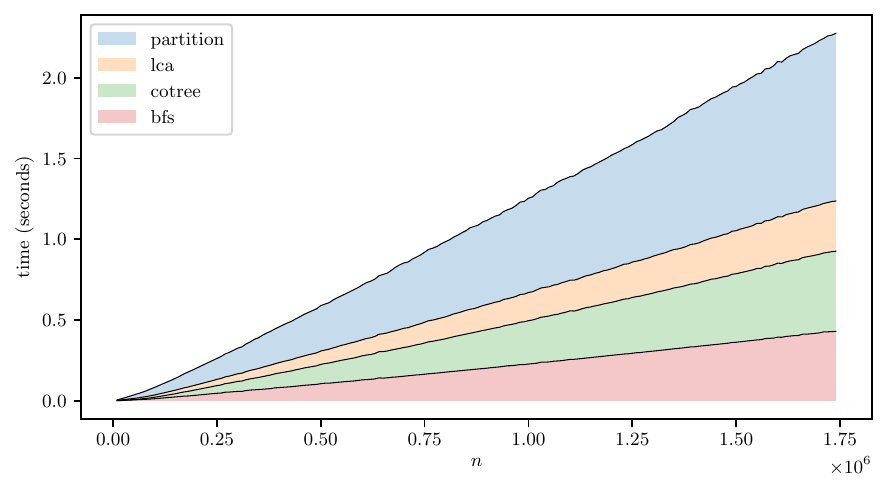} \\
  \includegraphics{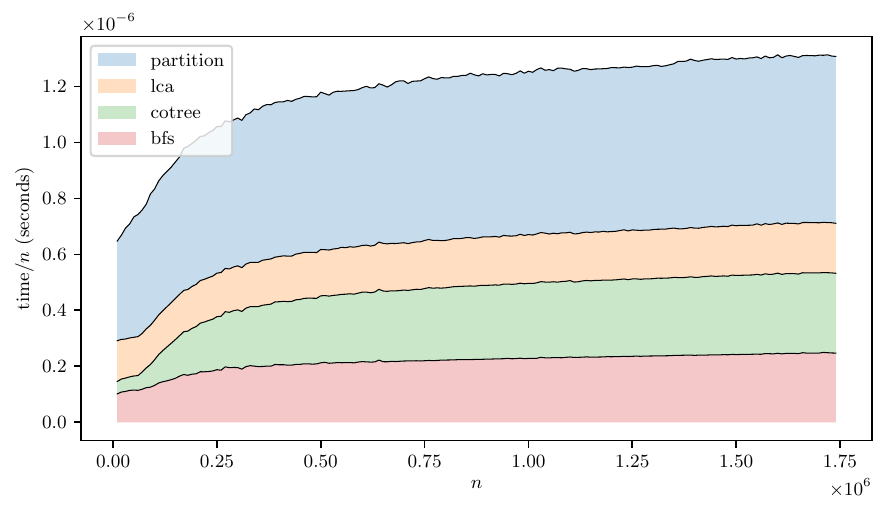}
  \caption{Running time of the tripod partition algorithm C++ implementation in \cite{morin:fatripar}}
  \label{plots}
\end{figure}

\Cref{plots} breaks down the total running time of the algorithms into the various stages, including the construction of the BFS tree $T$, the construction of the cotree $\overline{T}$, and the tripod partition algorithm proper. \Cref{plots} does not include the time required to read the input file that contains the triangulation $G$ and convert it into the triangle adjacency structure described before the proof of \cref{algorithm}. The upper part of \cref{plots} shows the total running time, in seconds. The lower part shows the total running time divided by the number of vertices.  For even the largest triangulations tested, the time per vertex is not more than $1.3\times 10^{-6}$ seconds ($1.3$ microseconds).  

The time per vertex appears to increase slowly with $n$ in the lower part of \cref{plots}. This is a caching effect. The CPU in this machine has three levels of cache, the largest of which has 6MB.  As the size of the input triangulation increases, the fraction of memory references that are available in the cache decreases, resulting in (slower) accesses to RAM.  In the limit, every memory access would be to RAM and, on this machine this would result in a processing time of approximately $1.4$ microseconds per vertex.

The input triangulations used to generate \cref{plots} were obtained by computing Delaunay triangulations of points uniformly distributed in a circle, with three additional points added to ensure that the outer face is a triangle.  The machine used to run the experiments is the third author's office desktop computer. The CPU in this machine is an \texttt{Intel(R) Core(TM) i5-4670K CPU @ 3.40GHz}. The machine has 16GB of DDR3 memory in four 4GB modules configured at a speed of 1600MT/s.  (This machine is nothing special. It was purchased in 2013 and is a standard desktop computer.)

The short summary of these results is that the implementation is fast in real terms on basic hardware. For comparison, the time required to make a $1$-million vertex triangulation using the (very fast) \texttt{delaunator} algorithm \cite{delfrrr:delaunator} is approximately 8 seconds. The time required to compute a tripod partition of this triangulation is approximately 1 second.  

\section*{Acknowledgement}

This research was initiated at the BIRS 21w5235 Workshop on Graph Product Structure Theory, held November 21--26, 2021 at the Banff International Research Station.  The authors are grateful to the workshop organizers and participants for providing a stimulating research environment.  We are especially grateful to Vida~Dujmović for sharing \cref{meta}.\ref{four_tree} with us.

\bibliographystyle{plainurlnat}
\bibliography{ps2}

\begin{thebibliography}{27}
\providecommand{\natexlab}[1]{#1}
\providecommand{\url}[1]{\texttt{#1}}
\providecommand{\urlprefix}{URL }
\expandafter\ifx\csname urlstyle\endcsname\relax
  \providecommand{\doi}[1]{\href{https://dx.doi.org/#1}{\nolinkurl{doi:#1}}}\else
  \providecommand{\doi}[1]{\href{https://dx.doi.org/#1}{\nolinkurl{doi:#1}}}\fi
\providecommand{\eprint}[2][]{\url{#2}}

\bibitem[{Alstrup et~al.(2004)Alstrup, Gavoille, Kaplan, and
  Rauhe}]{alstrup.gavoille.ea:nearest}
Stephen Alstrup, Cyril Gavoille, Haim Kaplan, and Theis Rauhe.
\newblock Nearest common ancestors: {A} survey and a new algorithm for a
  distributed environment.
\newblock \emph{Theory Comput. Syst.}, 37(3):441--456, 2004.
\newblock \doi{10.1007/s00224-004-1155-5}.

\bibitem[{Baumgart(1975)}]{baumgart.ea:vision}
Bruce~G. Baumgart.
\newblock A polyhedron representation for computer vision.
\newblock In \emph{Proceedings of the May 19-22, 1975, National Computer
  Conference and Exposition}, AFIPS '75, page 589–596. Association for
  Computing Machinery, New York, NY, USA, 1975.
\newblock \doi{10.1145/1499949.1500071}.

\bibitem[{Bekos et~al.(2024)Bekos, Lozzo, Hlin{\v{e}}n{\'y}, and
  Kaufmann}]{bekos.dalozzo.ea:graph}
Michael~A. Bekos, Giordano~Da Lozzo, Petr Hlin{\v{e}}n{\'y}, and Michael
  Kaufmann.
\newblock Graph product structure for {$h$}-framed graphs.
\newblock \emph{The Electronic Journal of Combinatorics}, 31(4):P4.56, 2024.
\newblock \doi{10.37236/12123}.

\bibitem[{Bender and Farach{-}Colton(2000)}]{bender.farach-colton:lca}
Michael~A. Bender and Martin Farach{-}Colton.
\newblock The {LCA} problem revisited.
\newblock In Gaston~H. Gonnet, Daniel Panario, and Alfredo Viola, editors,
  \emph{{LATIN} 2000: Theoretical Informatics, 4th Latin American Symposium,
  Punta del Este, Uruguay, April 10-14, 2000, Proceedings}, volume 1776 of
  \emph{Lecture Notes in Computer Science}, pages 88--94. Springer, 2000.
\newblock \doi{10.1007/10719839\_9}.

\bibitem[{Berkman and Vishkin(1993)}]{berman.vishkin:recursive}
Omer Berkman and Uzi Vishkin.
\newblock Recursive star-tree parallel data structure.
\newblock \emph{{SIAM} J. Comput.}, 22(2):221--242, 1993.
\newblock \doi{10.1137/0222017}.

\bibitem[{Bilonenko(2018)}]{delfrrr:delaunator}
Vladi Bilonenko.
\newblock delaunator-cpp: A really fast {C++} library for delaunay
  triangulation of 2d points.
\newblock \url{https://github.com/delfrrr/delaunator-cpp/}, 2018.
\newblock GitHub repository.

\bibitem[{Bose et~al.(2020)Bose, Dujmović, Javarsineh, and
  Morin}]{bose.dujmovic.ea:asymptotically}
Prosenjit Bose, Vida Dujmović, Mehrnoosh Javarsineh, and Pat Morin.
\newblock Asymptotically optimal vertex ranking of planar graphs.
\newblock \emph{CoRR}, abs/2007.06455, 2020.
\newblock \eprint{2007.06455}.

\bibitem[{Bose et~al.(2022)Bose, Morin, and Odak}]{bose.morin.ea:linear}
Prosenjit Bose, Pat Morin, and Saeed Odak.
\newblock An optimal algorithm for product structure in planar graphs.
\newblock In Artur Czumaj and Qin Xin, editors, \emph{18th Scandinavian
  Symposium and Workshops on Algorithm Theory, {SWAT} 2022, T{\'{o}}rshavn,
  Faroe Islands, June 27-29, 2022}, volume 227 of \emph{LIPIcs}, pages
  19:1--19:14. Schloss Dagstuhl - Leibniz-Zentrum f{\"{u}}r Informatik, 2022.
\newblock \doi{10.4230/LIPICS.SWAT.2022.19}.

\bibitem[{Diestel(2017)}]{diestel:graph}
Reinhard Diestel.
\newblock \emph{Graph Theory, Fifth Edition}, volume 173 of \emph{Graduate
  texts in mathematics}.
\newblock Springer, 2017.
\newblock \doi{10.1007/978-3-662-53622-3}.

\bibitem[{Dujmovi\'{c} et~al.(2020)Dujmovi\'{c}, Joret, Micek, Morin, Ueckerdt,
  and Wood}]{dujmovic.joret.ea:planar}
Vida Dujmovi\'{c}, Gwena{\"{e}}l Joret, Piotr Micek, Pat Morin, Torsten
  Ueckerdt, and David~R. Wood.
\newblock Planar graphs have bounded queue-number.
\newblock \emph{J. {ACM}}, 67(4):22:1--22:38, 2020.
\newblock \doi{10.1145/3385731}.

\bibitem[{Dujmovi{\'{c}} et~al.(2023)Dujmovi{\'{c}}, Morin, and
  Wood}]{dujmovic.morin.ea:graph}
Vida Dujmovi{\'{c}}, Pat Morin, and David~R. Wood.
\newblock Graph product structure for non-minor-closed classes.
\newblock \emph{Journal of Combinatorial Theory, Series {B}}, 162:34--67, 2023.
\newblock \doi{10.1016/J.JCTB.2023.03.006}.

\bibitem[{Dujmović et~al.(2021)Dujmović, Esperet, Gavoille, Joret, Micek, and
  Morin}]{dujmovic.esperet.ea:adjacency}
Vida Dujmović, Louis Esperet, Cyril Gavoille, Gwena{\"{e}}l Joret, Piotr
  Micek, and Pat Morin.
\newblock Adjacency labelling for planar graphs (and beyond).
\newblock \emph{J. {ACM}}, 68(6):42:1--42:33, 2021.
\newblock \doi{10.1145/3477542}.

\bibitem[{Dujmović et~al.(2019)Dujmović, Esperet, Joret, Walczak, and
  Wood}]{dujmovic.esperet.ea:planar}
Vida Dujmović, Louis Esperet, Gwena{\"{e}}l Joret, Bartosz Walczak, and
  David~R. Wood.
\newblock Planar graphs have bounded nonrepetitive chromatic number.
\newblock \emph{CoRR}, abs/1904.05269, 2019.
\newblock \eprint{1904.05269}.

\bibitem[{Dębski et~al.(2020)Dębski, Felsner, Micek, and
  Schröder}]{debski.felsner.ea:improved}
Michał Dębski, Stefan Felsner, Piotr Micek, and Felix Schröder.
\newblock Improved bounds for centered colorings.
\newblock In Shuchi Chawla, editor, \emph{Proceedings of the 2020 {ACM-SIAM}
  Symposium on Discrete Algorithms, {SODA} 2020, Salt Lake City, UT, USA,
  January 5-8, 2020}, pages 2212--2226. {SIAM}, 2020.
\newblock \doi{10.1137/1.9781611975994.136}.

\bibitem[{Esperet et~al.(2023)Esperet, Joret, and
  Morin}]{esperet.joret.ea:sparse}
Louis Esperet, Gwena{\"{e}}l Joret, and Pat Morin.
\newblock Sparse universal graphs for planarity.
\newblock \emph{Journal of the London Mathematical Society}, 108(4):1333--1357,
  2023.
\newblock \doi{10.1112/jlms.12781}.

\bibitem[{Fischer and Heun(2006)}]{fischer.heun:theoretical}
Johannes Fischer and Volker Heun.
\newblock Theoretical and practical improvements on the rmq-problem, with
  applications to {LCA} and {LCE}.
\newblock In Moshe Lewenstein and Gabriel Valiente, editors,
  \emph{Combinatorial Pattern Matching, 17th Annual Symposium, {CPM} 2006,
  Barcelona, Spain, July 5-7, 2006, Proceedings}, volume 4009 of \emph{Lecture
  Notes in Computer Science}, pages 36--48. Springer, 2006.
\newblock \doi{10.1007/11780441\_5}.

\bibitem[{Gavril(1974)}]{gavril:intersection}
Fănică Gavril.
\newblock The intersection graphs of subtrees in trees are exactly the chordal
  graphs.
\newblock \emph{Journal of Combinatorial Theory, Series B}, 16:47–56, 1974.
\newblock \doi{doi:10.1016/0095-8956(74)90094-X}.

\bibitem[{Gouin(2026)}]{gouin:planar}
K.~Gouin.
\newblock planar\_prod\_struct\_theorem.
\newblock \url{https://github.com/kgouin/planar_prod_struct_theorem}, 2026.
\newblock GitHub repository.

\bibitem[{Guibas and Stolfi(1985)}]{guibas.stolfi.ea:voronoi}
Leonidas Guibas and Jorge Stolfi.
\newblock Primitives for the manipulation of general subdivisions and the
  computation of voronoi.
\newblock \emph{ACM Trans. Graph.}, 4(2):74–123, 1985.
\newblock \doi{10.1145/282918.282923}.

\bibitem[{Harel and Tarjan(1984)}]{harel.tarjan:fast}
Dov Harel and Robert~Endre Tarjan.
\newblock Fast algorithms for finding nearest common ancestors.
\newblock \emph{{SIAM} J. Comput.}, 13(2):338--355, 1984.
\newblock \doi{10.1137/0213024}.

\bibitem[{Kettner(1999)}]{kettner.ea:visualization}
Lutz Kettner.
\newblock \emph{Software design in computational geometry and contour-edge
  based polyhedron visualization}.
\newblock Ph.D. thesis, ETH Z{\"u}rich, 1999.

\bibitem[{Morin(2021)}]{morin:fast}
Pat Morin.
\newblock A fast algorithm for the product structure of planar graphs.
\newblock \emph{Algorithmica}, 83(5):1544--1558, 2021.
\newblock \doi{10.1007/s00453-020-00793-5}.

\bibitem[{Morin(2026)}]{morin:fatripar}
Pat Morin.
\newblock fatripar: Linear-time tripod partitions of triangulations.
\newblock \url{https://github.com/patmorin/fatripar}, 2026.
\newblock GitHub repository.

\bibitem[{Pilipczuk and Siebertz(2021)}]{pilipczuk.siebertz:polynomial}
Michal Pilipczuk and Sebastian Siebertz.
\newblock Polynomial bounds for centered colorings on proper minor-closed graph
  classes.
\newblock \emph{J. Comb. Theory {B}}, 151:111--147, 2021.
\newblock \doi{10.1016/J.JCTB.2021.06.002}.

\bibitem[{Schieber and Vishkin(1988)}]{shieber.vishkin:on}
Baruch Schieber and Uzi Vishkin.
\newblock On finding lowest common ancestors: Simplification and
  parallelization.
\newblock \emph{{SIAM} J. Comput.}, 17(6):1253--1262, 1988.
\newblock \doi{10.1137/0217079}.

\bibitem[{Ueckerdt et~al.(2022)Ueckerdt, Wood, and
  Yi}]{ueckerdt.wood.ea:improved}
Torsten Ueckerdt, David~R. Wood, and Wendy Yi.
\newblock An improved planar graph product structure theorem.
\newblock \emph{The Electronic Journal of Combinatorics}, 29(2):P2.51, 2022.
\newblock \doi{10.37236/10614}.

\bibitem[{Weiler(1985)}]{weiler.ea:structure}
Kevin Weiler.
\newblock Edge-based data structures for solid modeling in curved-surface
  environments.
\newblock \emph{IEEE Computer Graphics and Applications}, 5(1):21--40, 1985.
\newblock \doi{10.1109/MCG.1985.276271}.

\end{thebibliography}

\end{document}